\documentclass[11pt]{article}

\usepackage{epsfig}

\newtheorem{thm}{Theorem}

\usepackage{amsfonts}
\usepackage{mathrsfs}
\usepackage{amssymb}
\usepackage{stmaryrd}
\usepackage{amsmath}
\usepackage{graphicx}
\usepackage{subfig}
\usepackage{color}
\usepackage{ifpdf}

\newenvironment{proof}{\makebox[7ex][l]{\it Proof:\/}}{\hfill\/ \hfill\/ {\it
Q.E.D.} \vspace{0.5ex}\\}

\begin{document}

\makeatletter

\title{\Large \bf Hybrid Control Technique for Switched LPV Systems and Its Application to Active Magnetic Bearing System}
% A Hybrid Control Framework for Switched LPV Systems with Application to Active Magnetic Bearings

\author{Fen Wu\thanks{Email: {\tt fwu@ncsu.edu}, Tel: (919) 515-5268.}
\vspace*{0.1in} \\
Dept. of Mechanical and Aerospace Engineering \\
North Carolina State University, Raleigh, NC 27695, USA}

\date{February 21, 2026}

\maketitle

\begin{abstract}
This paper proposes a novel hybrid control framework for switched linear parameter-varying (LPV) systems under hysteresis switching logic.
By introducing a controller state-reset mechanism, the
hybrid LPV synthesis problem is reformulated as a convex optimization problem expressed in terms of linear matrix inequalities (LMIs), enabling efficient computation of both switching LPV controller gains and reset matrices.
The proposed approach is then applied to active magnetic bearing (AMB) systems, whose rotor dynamics exhibit strong dependence on rotational speed.
Conventional LPV designs are often conservative due to large speed variations.
The proposed hybrid gain-scheduled controller explicitly accounts for bounds on parameter variation rates, employs multiple LPV controllers over distinct operating regions, and uses hysteresis switching to reduce chattering and ensure stability.
The effectiveness of the approach is demonstrated through a detailed AMB control design example.
%   This design procedure can also be extended to
%   high-performance control of turbo-machinery problems.
\end{abstract}

{\bf Keywords:} Hybrid control, switched LPV systems, hysteresis switching logic, controller state reset, LMI optimization.

\section{Introduction}
\label{sec:intro}

Hybrid systems have received sustained attention in the control community due to their ability to model and control complex dynamical behaviors that cannot be adequately addressed by purely continuous or discrete frameworks.
By integrating continuous dynamics with logical decision mechanisms, hybrid control architectures can effectively handle large uncertainties, nonlinearities, and mode-dependent behaviors, often alleviating fundamental limitations of conventional control designs and enabling improved performance and flexibility \cite{BraBM.TAC98,Ant.IEEE00,McCK.IEEE00,SchS.B00,GoeST.B12}.
As an important subclass of hybrid systems, switched linear systems which consisting of multiple linear subsystems and a rule governing switching among them, have been extensively studied \cite{LiberzonB03,SunG.B05,LinA.TAC09}.
Despite the availability of powerful linear analysis tools, stability and stabilization of switched systems remain challenging, since instability may arise even when all individual subsystems are asymptotically stable \cite{LibM.CSM99}.
Consequently, the switching signal plays a critical role in determining system behavior and must be explicitly considered in controller design.

For autonomous switching, stabilization can be treated as a robust control problem, where the existence of a common Lyapunov function (CLF) guarantees stability under arbitrary switching \cite{LiberzonB03}.
However, CLF-based conditions are often overly conservative, particularly when specific switching logics are prescribed. This has motivated the development of multiple Lyapunov function (MLF) approaches \cite{Bra.TAC98,YeMH.TAC98}, which enable less conservative stability analysis under structured switching schemes such as dwell-time, average dwell-time, and hysteresis switching \cite{Mor.B97,Bra.TAC98,YeMH.TAC98,LibM.CSM99}.
Among these, hysteresis switching is particularly attractive as it implicitly enforces dwell-time, prevents chattering near switching boundaries, and improves robustness to noise and uncertainty.

While MLF-based methods have been successfully applied to state-feedback synthesis using LMI formulations, output-feedback control of switched systems remains substantially more difficult. In particular, enforcing boundary conditions between Lyapunov functions typically leads to nonconvex bilinear matrix inequalities (BMIs), which are computationally challenging.
Existing approaches either decouple controller synthesis from switching conditions, resulting in conservative performance or excessive dwell-time requirements or rely on restrictive assumptions such as full state availability at switching instants \cite{LuW.Au04,LuW.ACC04}.
More recently, hybrid output-feedback frameworks incorporating controller-state reset mechanisms have been developed for switched systems under dwell-time switching logic.
This enabled convex LMI-based synthesis while explicitly accounting for switching-induced boundary conditions, significantly reducing conservatism and improving performance \cite{YuaW2015,YuaLWD2016}.

Parallel to these developments, hysteresis switching logic has been widely studied as a unifying mechanism in adaptive and hybrid control.
Early work by Morse, Mayne, and Goodwin \cite{MorMG92} demonstrated how hysteresis-based switching can prevent chattering and ensure robustness in parameter adaptive control. Subsequent research extended hysteresis logic to switched and linear parameter-varying (LPV) systems \cite{LuW.Au04}, where it has been shown to improve stability margins, enforce implicit dwell-time constraints, and enhance robustness to modeling uncertainty and measurement noise.
Recent efforts further integrate hysteresis switching within hybrid inclusion frameworks and optimization-based controller synthesis \cite{GoeST.B12}.

Active magnetic bearings (AMBs) provide non-contact electromagnetic support for rotors in high-speed rotating machinery and are increasingly used in aerospace, energy storage, and precision industrial applications.
Their advantages include the elimination of mechanical friction and lubrication, reduced wear, and the capability to actively suppress synchronous vibrations due to imbalance and shaft runout.
However, AMBs are inherently open-loop unstable, and controller design is further complicated by strong coupling effects, gyroscopic forces, and rotor flexibility especially at high rotational speeds.

To address these challenges, gain-scheduling and LPV control techniques have been widely investigated for AMB systems. Early LPV approaches focused on rigid-rotor models with rotor speed as the scheduling parameter, revealing significant conservatism when large speed variations are present \cite{TsiM96,SivN96}. Flexible rotor effects were later addressed using robust control and $\mathcal{H}_\infty$ or $\mu$-synthesis methods, often at the expense of neglecting time-varying gyroscopic effects for tractability \cite{FujHM93,NonI96}.
Switching LPV control frameworks were subsequently proposed to reduce conservatism by partitioning the parameter space and designing locally optimized controllers, with stability ensured via Lyapunov continuity across parameter boundaries \cite{Wu2001}. More recent work has explored gain-scheduled and LPV controller implementation using Youla parametrization and smooth switching or interpolation techniques to guarantee stability during controller transitions \cite{BalWS2012,AtoSMM2022}.

Despite these advances, existing LPV and switching-based AMB controllers often suffer from conservatism when large parameter variations are present, or impose restrictive stability constraints that limit achievable performance.
In particular, the systematic integration of hysteresis switching logic and hybrid controller-state resets into LPV synthesis remains underexplored.

Motivated by the development of hybrid control for switched system under dwell-time logic \cite{YuaW2015},
we propose a hybrid LPV control framework with hysteresis switching logic for robust gain-scheduling control in this paper.
By combining switching LPV controllers with a controller state-reset mechanism, the proposed approach guarantees closed-loop stability under hysteresis switching while leading to convex LMI-based synthesis conditions.
This enables simultaneous design of local LPV controllers and reset matrices, achieving improved performance with reduced conservatism.
The effectiveness of the proposed method is demonstrated through an active magnetic bearing application, where the resulting controller accommodates wide variations in rotor speed and achieves improved control performance, particularly in the high-speed operating regime.

The notation in this paper is standard.
The sets of real and positive numbers are denoted by ${\bf R}$ and ${\bf R}_+$, respectively.
The set of real $m \times n$ matrices is ${\bf R}^{m\times n}$, and the transpose of a real matrix $M$ is written as $M^T$.
For real matrices, the Hermitian operator ${\it He} \{\cdot\}$ is defined by ${\it He} \{M\} = M + M^T$.
The identity matrix of appropriate dimension is denoted by $I$.
The sets of real symmetric and symmetric positive definite $n\times n$ matrices are denoted by ${\bf S}^{n\times n}$
and ${\bf S}^{n\times n}_+$, respectively.
For $M \in{\bf S}^{n\times n}$, the notation $M > 0$ ($M
\geq 0$) indicates that $M$ is positive definite (positive
semi-definite), while $M < 0$ ($M \leq 0$) denotes negative
definite (negative semi-definite) matrices.
The symbol $\star$ in LMIs represents entries implied by
symmetry.
For $x \in {\bf R}^{n}$, the Euclidean norm is defined as $\|x\| := (x^T x)^{1/2}$.
The space of square-integrable functions is denoted by
${\cal L}_2$; for any $u\in {\cal L}_2$, $\| u \|_2
= [\int^{\infty}_{0} u^T(t) u(t) dt]^{1/2}$.
Finally, for integers $k_1 < k_2$, ${\bf I}[k_1,k_2] =
\{k_1, k_1+1, \cdots, k_2\}$.

The remainder of this paper is organized as follows.
Section \ref{sec:hyblpv} presents the formulation of the hybrid LPV control problem and derives the corresponding LMI-based synthesis conditions.
Section \ref{sec:amb} applies the proposed approach to active magnetic bearing systems and develops hybrid LPV controllers that achieve improved performance while guaranteeing stability under hysteresis switching.
Finally, Section \ref{sec:conclusion} concludes the paper.

\section{Hybrid LPV Control Approach}
\label{sec:hyblpv}

Consider an open-loop switched LPV system whose state-space matrices depend on a scheduling parameter $\rho$.
It is assumed that $\rho$ belongs to a compact set ${\cal P} \subset {\bf R}^s$ and that its variation rate satisfies
\[
{\cal V} = \left\{ \nu: \underline{\nu}_k \leq \dot{\rho}_k \leq \bar{\nu}_k, k = 1,2, \dots, s \right\},
\]
where ${\cal V}$ is a convex polytope containing the origin.

Suppose that the parameter set ${\cal P}$ is covered by a finite collection of closed subsets ${\cal P}_i, i \in Z_{N_p}$, separated by a family of switching surfaces ${\cal S}_{ij}$, with index set $Z_{N_p} = \left\{1, 2, \ldots, N_p \right\}$ and $\mathcal{P} = \bigcup \mathcal{P}_i$.
Adjacent parameter subsets have overlapping interiors and are separated by switching surfaces.
Within each parameter subset ${\cal P}_i$, the system dynamics are described by
\begin{align}
\begin{bmatrix}
\dot{x}(t) \\
z(t) \\
y(t)
\end{bmatrix} & = \begin{bmatrix}
A_i(\rho(t)) & B_{1,i}(\rho(t)) & B_{2,i}(\rho(t)) \\
C_{1,i}(\rho(t)) & D_{11,i}(\rho(t)) & D_{12,i}(\rho(t)) \\
C_{2,i}(\rho(t)) & D_{21,i}(\rho(t)) & D_{22,i}(\rho(t))
\end{bmatrix}
\begin{bmatrix}
x(t) \\
w(t) \\
u(t)
\end{bmatrix}
\label{eqn:lpvolp}
\end{align}
where $x \in {\bf R}^n$ is the plant state, $e \in {\bf R}^{n_e}$ the controlled output, $d \in {\bf R}^{n_d}$ the disturbance input, $y \in {\bf R}^{n_y}$ the measured output, and $u \in {\bf R}^{n_u}$ the control input.
All state-space matrices are continuous functions of $\rho$.
Throughout the paper, the following standard assumptions are imposed:
\begin{description}
\item[(A1)]
The triple $(A_i(\rho), B_{2,i}(\rho), C_{2,i}(\rho))$ is parameter-dependently stabilizable and detectable for all $\rho$.
\item[(A2)]
The matrices $D_{12,i}(\rho)$ and $D_{21,i}(\rho)$ have full column and row rank, respectively.
\item[(A3)]
$D_{11,i}(\rho) = 0$ and $D_{22,i}(\rho) = 0$.
\end{description}

The control objective is to construct a hybrid LPV control law with hysteresis switching logic that asymptotically stabilizes the switched LPV plant (\ref{eqn:lpvolp}) while achieving optimal $\mathcal{L}_2$-gain performance.
The synthesis conditions are required to be expressed in terms of linear matrix inequalities (LMIs).
In the context of the active magnetic bearing (AMB) application, the control objective further includes disturbance rejection, gyroscopic compensation, and automatic balancing, which can be formulated as the minimization of rotor displacement from its centerline in the presence of unknown torque disturbances.

When hysteresis switching is employed, any two adjacent parameter subsets are required to overlap, as shown in Fig. \ref{fig:hystere}. Consequently, two switching surfaces exist between adjacent subsets. The surface ${\cal S}_{ij}$ specifies a one-directional transition from ${\cal P}_i$ to ${\cal P}_j$, thereby preventing chattering and ensuring well-posed switching behavior.

\begin{figure}[!htb]
\centering
%   \centering\psfig{file=epsfile/hystere.eps,width=3.0in,height=1.45in}
\includegraphics[width=0.5\hsize]{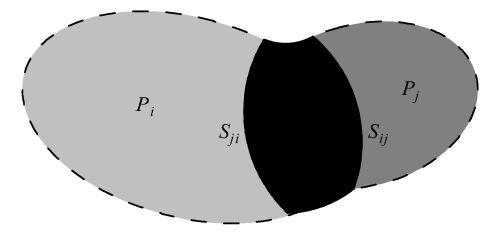}
\caption{Hysteresis switching regions.}
\label{fig:hystere}
\end{figure}

The switching event occurs when the parameter trajectory hits one of the switching surfaces $\mathcal{S}_{ij}$ or $\mathcal{S}_{ji}$.
The evolution of the switching signal $\sigma$ is described as
follows: Let $\sigma(0) = i$ if $\rho(0) \in \mathcal{P}_i$. For each $t > 0$, if $\sigma(t^-) = i$ and $\rho(t) \in \mathcal{P}_i$, keep $\sigma(t) = i$.
On the other hand, if $\sigma(t^-) = i$ but $\rho(t) \in \mathcal{P}_j \cap \mathcal{P}_i$, i.e., hitting the switching
surface $\mathcal{S}_{ij}$, let $\sigma(t) = j$.
%   Repeating this procedure, we generate a piecewise constant
%   signal $\sigma$ which is continuous from the right
%   everywhere.

The family of hybrid LPV controllers is given by
\begin{align}
\begin{aligned}
\begin{bmatrix}
\dot{x}_k \\
u
\end{bmatrix} &= \begin{bmatrix}
A_{k,i}(\rho,\dot{\rho}) & B_{k,i}(\rho) \\
C_{k,i}(\rho) & D_{k,i}(\rho)
\end{bmatrix}
\begin{bmatrix}
x_{k} \\
y
\end{bmatrix}, \qquad i \in Z_{N_p} \\
x_k^+ &= \Delta_{ij}(\rho) x_k, \quad \mbox{when switching occurs}
\end{aligned}
\label{eqn:hyblpvctrl}
\end{align}
where each controller is designed for a specific parameter subset ${\cal P}_i$, and $x_k \in {\bf R}^{n_k}$ denotes the controller state.
Control objectives may differ-and may even be conflicting-across parameter regions.
Each local LPV controller is designed to achieve optimal performance within its corresponding subset while preserving closed-loop stability under the hysteresis switching logic.

As illustrated in Fig. \ref{Fig.swHyb}, the proposed architecture consists of a switching dynamic output-feedback LPV controller and a supervisory logic enforcing controller state reset rules through an embedded hybrid loop (denoted by the gray block ${\cal H}$).
During flow intervals, the system evolves continuously under the active controller.
At switching instants, the controller state undergoes a discrete jump.
%   This structure is motivated by reset control and hybrid
%   dynamical system theory \cite{GoeST.B12}.

%%%%%%%%%%%%%%%%%%%%%%%%%%%%%%%%%%%%%%%%%%%%%%%%
\begin{figure}[!htb]
\centering
\includegraphics[scale=0.35]{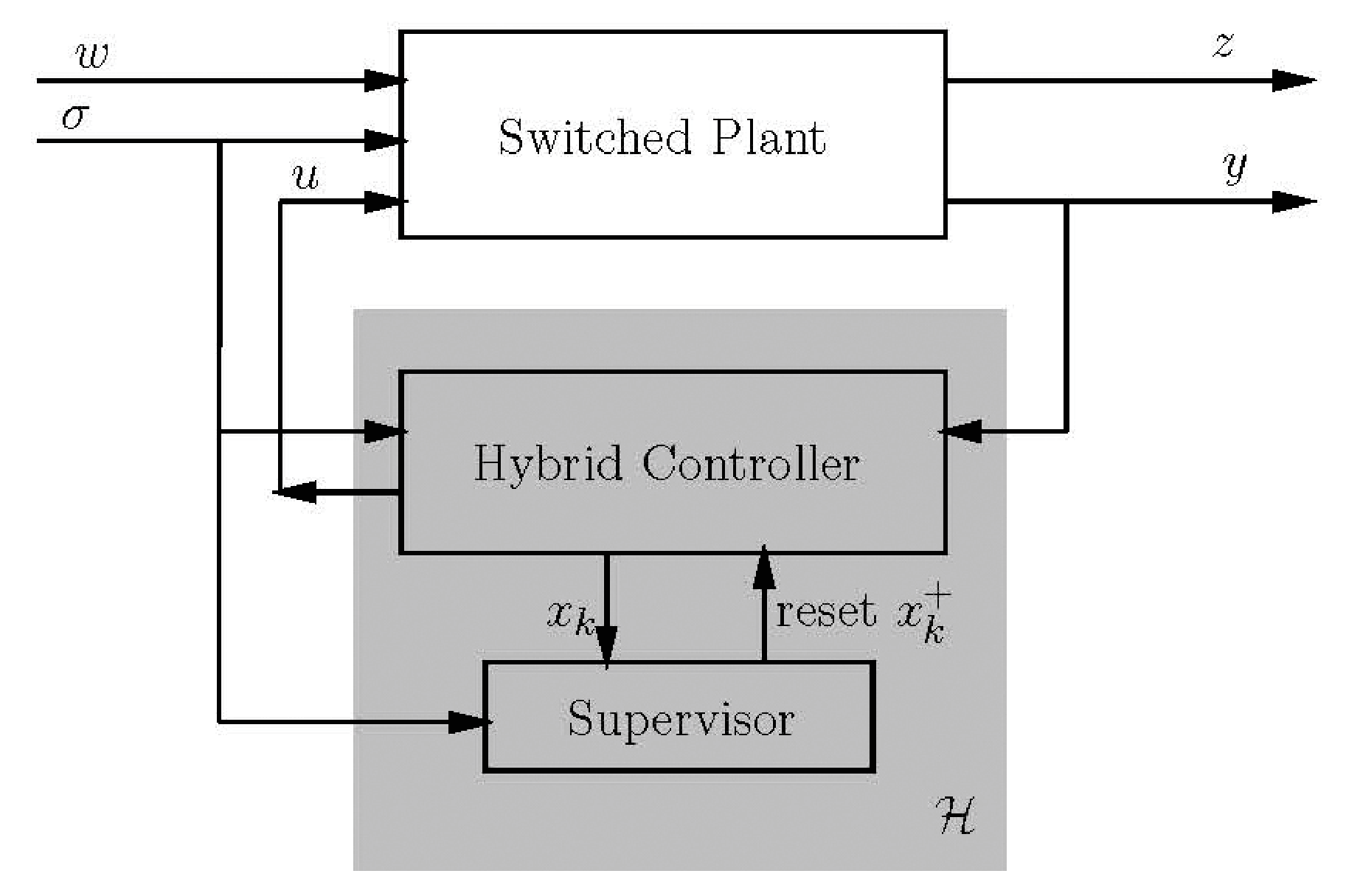}
\caption{The proposed hybrid control scheme.}
\label{Fig.swHyb}
\end{figure}
%%%%%%%%%%%%%%%%%%%%%%%%%%%%%%%%%%%%%%%%%%%%%%%%%

Under the proposed hybrid LPV controller, the closed-loop system takes the form
\begin{align}
\begin{aligned}
\begin{bmatrix}
\dot{x}_{cl} \\
z
\end{bmatrix} &= \begin{bmatrix}
A_{cl,i}(\rho,\dot{\rho}) & B_{cl,i}(\rho) \\
C_{cl,i}(\rho) & D_{cl,i}(\rho)
\end{bmatrix}
\begin{bmatrix}
x_{cl} \\
w
\end{bmatrix} \\
x_{cl}^+ & = A_{r,ij}(\rho) x_{cl}, \quad \mbox{when switching occurs}
\end{aligned}
\label{eqn:lpvclp}
\end{align}
where $x^T_{cl} = \left[ x^T \ \ x^T_k \right] \in {\bf
R}^{2n}$ and
\begin{align*}
A_{cl,i}(\rho) & = \begin{bmatrix}
A_{i}(\rho) + B_{2,i}(\rho) D_{k,i}(\rho) C_{2,i}(\rho) & B_{2,i}(\rho) C_{k,i}(\rho) \\
B_{k,i}(\rho) C_{2,i}(\rho) & A_{k,i}(\rho)
\end{bmatrix} \\
B_{cl,i}(\rho) &= \begin{bmatrix}
B_{1,i}(\rho) + B_{2,i}(\rho) D_{k,i}(\rho) D_{21,i}(\rho) \\
B_{k,i}(\rho) D_{21,i}(\rho)
\end{bmatrix} \\
C_{cl,i}(\rho) &= \begin{bmatrix}
C_{1,i}(\rho) + D_{12,i}(\rho) D_{k,i}(\rho) C_{2,i}(\rho) & D_{12,i}(\rho) C_{k,i}(\rho)
\end{bmatrix} \\
D_{cl,i}(\rho) &= \begin{bmatrix}
D_{11,i}(\rho) + D_{12,i}(\rho) D_{k,i}(\rho) D_{21,i}(\rho)
\end{bmatrix} \\
A_{r,ij}(\rho) & = \begin{bmatrix}
I_n & 0 \\
0 & \Delta_{ij}(\rho)
\end{bmatrix}.
\end{align*}
The resulting closed-loop dynamics constitute a switched LPV system with possible state discontinuities and controller gain changes at switching surfaces.

Let $X_i(\rho)$ denote a positive-definite matrix function associated with the Lyapunov function $V_i(x_{cl},\rho) = x_{cl}^T X_i(\rho) x_{cl}$ when the $i$th controller is
active.
On the switching surface $\mathcal{S}_{ij}$, the hysteresis switching condition requires
\begin{equation}
\label{eqn:hystana_LMI1}
V_j(x_{cl},\rho) \leq V_i(x_{cl},\rho)
\end{equation}
which guarantees that the Lyapunov function does not increase during switching from ${\cal P}_i$ to ${\cal P}_j$. Consequently, the $j$th controller can be safely activated.
Within each parameter subset ${\cal P}_i$, the closed-loop system is required to satisfy a bounded real condition with performance level $\gamma_i$:
\begin{align}
& \begin{bmatrix}
\begin{matrix}
A^T_{cl,i}(\rho) X_i(\rho) + X_i(\rho) A_{cl,i}(\rho) \\ + \displaystyle{\sum_{k=1}^s}
\left\{\underline{\nu}_k, \overline{\nu}_k \right\} \frac{\partial X_i}{\partial \rho_k} \end{matrix}
& X_i(\rho) B_{cl,i}(\rho) & C_{cl,i}^T(\rho) \\
B^T_{cl,i}(\rho) X_i(\rho) & -\gamma_i I_{n_d} & D_{cl,i}^T(\rho) \\
C_{cl,i}(\rho) & D_{cl,i}(\rho) & -\gamma_i I_{n_e}
\end{bmatrix} < 0 \label{eqn:hystana_LMI2}
\end{align}
Standard LPV results \cite{WuYPB96,HugW.B12} ensure that these conditions are automatically satisfied by solutions derived from conventional LPV analysis, implying that the proposed hybrid formulation is no more conservative than the single Lyapunov function approach.

Based on the hysteresis switching logic, sufficient synthesis conditions for hybrid LPV control are summarized in the following theorem.

\begin{thm}
\label{thm:hyst_syn}
Consider the open-loop switched LPV system (\ref{eqn:lpvolp}), the parameter set ${\cal P}$ and an overlapped covering $\left\{{\cal P}_i \right\}_{i \in Z_{N_p}}$.
Suppose there exist continuously differentiable, positive definite matrix functions $R_i(\cdot), S_i(\cdot):
\mathbf{R}^s \rightarrow \mathbf{S}^{n \times n}_+, i \in Z_{N_p}$, such
that for any $\rho \in \mathcal{P}_i$
\begin{align}
\mathcal{N}^T_R(\rho) \begin{bmatrix}
\begin{matrix} R_i(\rho) A_i^T(\rho) + A_i(\rho) R_i(\rho) \\
- \displaystyle{\sum_{k=1}^s} \left\{
\underline{\nu}_k, \overline{\nu}_k \right\} \frac{\partial R_i}{\partial \rho_k} \end{matrix}
& R_i(\rho) C_{i,1}^T(\rho) & B_{i,1}(\rho) \\
C_{i,1}(\rho) R_i(\rho) & -\gamma_i I_{n_e} & D_{i,11}(\rho) \\
B^T_{i,1}(\rho) & D^T_{i,11}(\rho) & -\gamma_i I_{n_d}
\end{bmatrix} \mathcal{N}_R(\rho) &< 0 \label{eqn:hystLMI1} \\
\mathcal{N}^T_S(\rho) \begin{bmatrix}
\begin{matrix} A_i^T(\rho) S_i(\rho) + S_i(\rho) A_i(\rho) \\
+ \displaystyle{\sum_{k=1}^s} \left\{
\underline{\nu}_k, \overline{\nu}_k \right\} \frac{\partial S_i}{\partial \rho_k} \end{matrix}
& S_i(\rho) B_{i,1}(\rho) & C_{i,1}^T(\rho) \\
B^T_{i,1}(\rho) S_i(\rho) & -\gamma_i I_{n_d} & D_{i,11}^T(\rho) \\
C_{i,1}(\rho) & D_{i,11}(\rho) & -\gamma_i I_{n_e}
\end{bmatrix} \mathcal{N}_S(\rho) &< 0 \label{eqn:hystLMI2} \\
\begin{bmatrix}
R_i(\rho) & I_{n} \\
I_{n} & S_i(\rho)
\end{bmatrix} &> 0 \label{eqn:hystLMI3}
\end{align}
holds with performance level $\gamma_i > 0$
and for $\rho \in \mathcal{S}_{ij}$, there exist matrices  $\hat{\Delta}_{ij}(\rho)$ satisfying the boundary condition
\begin{align}
\begin{bmatrix}
R_i(\rho) & \star & \star & \star \\
I_n & S_i(\rho) & \star & \star \\
R_i(\rho) & I_n & R_j(\rho) & \star \\
\hat{\Delta}_{ij}(\rho) & S_j(\rho) & I_n & S_j(\rho)
\end{bmatrix} \geq 0, \label{eqn:hystLMI4}
\end{align}
where
$\mathcal{N}_R(\rho) = \emph{Ker} \begin{bmatrix}
B^T_{i,2}(\rho) & D^T_{i,12}(\rho) & 0
\end{bmatrix}$ and $\mathcal{N}_S(\rho) = \emph{Ker} \begin{bmatrix}
C_{i,2}(\rho) & D_{i,21}(\rho) & 0
\end{bmatrix}$.
Then, under the hysteresis switching logic, the closed-loop switched LPV system (\ref{eqn:lpvclp}) is
asymptotically stable stable over the entire parameter set ${\cal P}$ when controlled by the hybrid LPV controllers  (\ref{eqn:hyblpvctrl}).
Moreover, the closed-loop system achieves the ${\cal L}_2$
gain  performance $\|z \|_2 < \gamma \|w \|_2, \gamma = \max \left\{ \gamma_i \right\}_{i \in Z_{N_p}}$.
\end{thm}

\begin{proof}
Consider the Lyapunov-like functions for the closed-loop system (\ref{eqn:lpvclp}) defined as
\begin{align*}
V_i(x_{cl}) &= x_{cl}^T X_{i}(\rho) x_{cl}, \qquad i \in {\bf I}[1, N_p]
\end{align*}
where each $V_i$ corresponds to the closed-loop subsystem associated with the $i$th parameter subset.

Let the switching instants over a finite interval $[0, T]$
be denoted by $0 = t_0, t_1, \ldots,
t_N$.
From the $(1,1)$ block of the inequality (\ref{eqn:hystana_LMI2})
and the compactness of each parameter subset $\mathcal{P}_i$, it follows that, for all $\rho \in {\cal P}_i$ with the parameter variation rate $\dot{\rho} \in [\underline{\nu}, \bar{\nu}]$,
\[
A^T_{cl,i}(\rho) X_i(\rho) + X_i(\rho) A_{cl,i}(\rho) + \sum_{k=1}^s \dot{\rho}_k
\frac{\partial X_i(\rho)}{\partial \rho_k} \leq -\lambda X_i(\rho)
\]
for some $\lambda > 0$.
Therefore, for any $t \in \left[t_k, t_{k+1}\right)$,
\[
V_\sigma(x_{cl}(t)) \leq e^{-\lambda (t-t_k)} V_\sigma(x_{cl}(t_k)).
\]
Under the hysteresis switching logic, the boundary condition
(\ref{eqn:hystana_LMI1}) guarantees that the Lyapunov function does not increase at switching instants, i.e.,  $V_\sigma(x_{cl}(t_k)) \leq
V_\sigma(x_{cl}(t^-_k))$.
By recursively applying the above inequalities, we obtain
\begin{align*}
V_\sigma(x_{cl}(t))
& \leq e^{-\lambda (t-t_k)} V_\sigma(x_{cl}(t^-_k)) \\
& \leq e^{-\lambda (t-t_k)} e^{-\lambda (t_k - t_{k-1})}
V_\sigma(x_{cl}(t_{k-1})) \\
& \vdots \\
& \leq e^{-\lambda t} V_\sigma(x_{cl}(0)),
\end{align*}
which establishes exponential stability of the closed-loop switched LPV system (\ref{eqn:lpvclp}).

To prove the ${\cal L}_2$-gain performance, consider zero  initial condition $x_{cl}(0) = 0$. From inequality (\ref{eqn:hystana_LMI2}), it follows that, within each parameter subset,
\[
\dot{V}_\sigma + \frac{1}{\gamma}z^T z - \gamma w^T w < 0, \qquad
\gamma = \max \left\{ \gamma_i \right\}_{i \in Z_N}.
\]
Integrating over $[0, T]$ yields
\[
V_\sigma(x_{cl}(T)) - V_\sigma(x_{cl}(0)) + \frac{1}{\gamma} \| z \|^2_2
- \gamma \| w \|^2_2 < 0.
\]
Since $V_\sigma(x_{cl}(T)) \geq 0$ and $V_\sigma(x_{cl}(0)) = 0$, the desired performance bound
$\| z \|_2 < \gamma \| w \|_2$ follows.

Next, define the matrices
\begin{align*}
Z_{1,i}(\rho) = \begin{bmatrix}
R_i(\rho) & I_n \\
M_i^T(\rho) & 0
\end{bmatrix},
\qquad
Z_{2,i}(\rho) = \begin{bmatrix}
I_n & S_i(\rho) \\
0 & N_i^T(\rho)
\end{bmatrix}
\end{align*}
such that $M_i(\rho) N_i^T(\rho) = I - R_i(\rho) S_i(\rho)$ and $X_{i}(\rho) Z_{1,i}(\rho) = Z_{2,i}(\rho)$.
%   Then, we have $X_i^{-1} = -N_i^T R_i M_i^{-T}$.
From condition (\ref{eqn:hystLMI3}), we have
\[
Z_{1,i}^T(\rho) X_{i}(\rho) Z_{1,i}(\rho) = \begin{bmatrix}
R_i(\rho) & I_n \\
I_n & S_i(\rho)
\end{bmatrix} > 0,
\]
which implies $X_{i}(\rho) > 0$ for all $i \in {\bf I}[1,N_p]$.

Applying the Elimination Lemma to condition (\ref{eqn:hystana_LMI2}) removes the controller matrices $A_{k,i}(\cdot), B_{k,i}(\cdot), C_{k,i}(\cdot)$ and $D_{k,i}(\cdot)$, yielding the equivalent LMIs
(\ref{eqn:hystLMI1})-(\ref{eqn:hystLMI2}).
%   multiplying matrix ${\it diag}\{Z_{1,i},I_{n_w},I_{n_z}\}$ %   to the right and its transpose from the left on both sides %   of inequality (\ref{Thm.LMI1*}), we have the following
%   results by congruent transformation \cite{SchGC.TAC97}:
%   \begin{align*}
%   Z_{1,i}^T P_i A_{cl,i} Z_{1,i}
%   &= Z_{2,i}^T A_{cl,i} Z_{1,i} = \\
%   & \hspace*{-0.3in} \begin{bmatrix}
%   A_{p,i} R_i + B_{p2,i} \hat{C}_{k,i} & A_{p,i} + B_{p2,i}
%   \hat{D}_{k,i} C_{p2,i} \\
%   \hat{A}_{k,i} & S_i A_{p,i} + \hat{B}_{k,i} C_{p2,i}
%   \end{bmatrix} \\
%   B_{cl,i}^T P_i Z_{1,i}
%   &= B_{cl,i}^T Z_{2,i} = \\
%   & \hspace*{-0.3in} \begin{bmatrix}
%   B_{p1,i}^T + D_{p21,i}^T \hat{D}_{k,i}^T B_{p2,i}^T &
%   B_{p1,i}^T S_i + D_{p21,i}^T \hat{B}_{k,i}^T
%   \end{bmatrix} \\
%   C_{cl,i} Z_{1,i} = \\
%   & \hspace*{-0.3in} \begin{bmatrix}
%   C_{p1,i} R_i + D_{p12,i} \hat{C}_{k,i} & C_{p1,i} +
%   D_{p12,i} \hat{D}_{k,i} C_{p2,i}
%   \end{bmatrix}
%   \end{align*}
%   where
%   \begin{align}
%   \begin{aligned}
%   \label{barABCD}
%   \hat{A}_{k,i} &= S_i (A_{p,i} + B_{p2,i} D_{k,i} C_{p2,i}) %   R_i + N_i B_{k,i} C_{p2,i} R_i \\
%   &\hspace*{0.25in} + S_i B_{p2,i} C_{k,i} M_i^T + N_i
%   A_{k,i} M_i^T \\
%   \hat{B}_{k,i} &= S_i B_{p2,i} D_{k,i} + N_i B_{k,i} \\
%   \hat{C}_{k,i} &= D_{k,i} C_{p2,i} R_i + C_{k,i} M_i^T \\
%   \hat{D}_{k,i} &= D_{k,i}.
%   \end{aligned}
%   \end{align}
%   Then, we arrive at conditions (\ref{LMI1}).

Finally, at a switching instant from ${\cal P}_i$ to ${\cal P}_j$, the post-jump Lyapunov function satisfies
$V_j(x_{cl}^+) = {x_{cl}^+}^T X_{j}(\rho) x_{cl}^+ = x_{cl}^T A_{r,ij}^T(\rho) X_{j}(\rho) A_{r,ij}(\rho) x_{cl}$.
The boundary condition (\ref{eqn:hystana_LMI1}) is equivalent to $X_{i}(\rho) - A_{r,ij}^T(\rho) X_{j}(\rho) A_{r,ij}(\rho) \geq 0$,
which, by Schur complement, can be written as
\begin{equation}
\label{eqn:boundary}
\begin{bmatrix}
X_{i}(\rho) & \star \\
X_{j}(\rho) A_{r,ij}(\rho) & X_{j}(\rho)
\end{bmatrix} \geq 0.
\end{equation}
Note that
\begin{align*}
Z_{1,j}^T(\rho) X_{j}(\rho) A_{r,ij}(\rho) Z_{1,i}(\rho)
&= Z_{2,j}^T(\rho) A_{r,ij}(\rho) Z_{1,i}(\rho) \\
&= \begin{bmatrix}
R_i(\rho) & I_n \\
\hat{\Delta}_{ij}(\rho) & S_j(\rho)
\end{bmatrix}
\end{align*}
where
\begin{equation}
\label{barDel}
\hat{\Delta}_{ij}(\rho) = S_j(\rho) R_i(\rho) + N_j(\rho) \Delta_{ij}(\rho) M_i^T(\rho).
\end{equation}
Pre- and post-multiplying by ${\it diag}\{Z_{1,i}^T(\rho),
Z_{1,j}^T(\rho)\}$ and its transpose on eqn. (\ref{eqn:boundary}) yields
\[
\begin{bmatrix}
R_i(\rho) & \star & \star & \star \\
I_n & S_i(\rho) & \star & \star \\
R_i(\rho) & I_n & R_j(\rho) & \star \\
\hat{\Delta}_{ij}(\rho) & S_j(\rho) & I_n & S_j(\rho)
\end{bmatrix} \geq 0.
\]
This establishes (\ref{eqn:hystLMI4}).
The controller reset matrix can then be recovered by inverting equation (\ref{barDel}).
\end{proof}

After solving for the matrix functions $R_i(\rho)$ and $S_i(\rho)$, the controller gains are reconstructed using the standard formulas in \cite{WuYPB96,HugW.B12}:
\begin{align}
\hspace*{-0.1in} A_{k,i}(\rho, \dot{\rho}) & = -N_i^{-1}(\rho) \left\{ A^T_i(\rho) - S_i(\rho) \frac{d R_i}{d t} - N_i(\rho) \frac{d M_i^T}{d t} \right. \nonumber \\
& \hspace*{-0.25in} + S_i(\rho) \left[ A_i(\rho) + B_{2,i}(\rho) F_i(\rho) +
L_i(\rho) C_{2,i}(\rho) \right] R_i(\rho) \nonumber \\
& \hspace*{-0.25in} + \frac{1}{\gamma_i}
S_i(\rho) \left[ B_{1,i}(\rho) + L_i(\rho) D_{21,i}(\rho) \right]
B^T_{1,i}(\rho) \nonumber \\
& \left. \hspace*{-0.25in} + \frac{1}{\gamma_i} C^T_{1,i}(\rho)
\left[ C_{1,i}(\rho) + D_{12,i}(\rho)
F_i(\rho) \right] R_i(\rho) \right\} M_i^{-T}(\rho) \label{eqn:lpv_aaaa} \\
B_{k,i}(\rho) & = N_i^{-1}(\rho) S_i(\rho) L_i(\rho) \label{eqn:lpv_bbbb} \\
C_{k,i}(\rho) & = F_i(\rho) R_i(\rho) M_i^{-T}(\rho) \label{eqn:lpv_cccc} \\
D_{k,i}(\rho) & = 0 \label{eqn:lpv_dddd} \\
\Delta_{ij} &= N_j^{-1}(\rho) \left(\hat{\Delta}_{ij} - S_j(\rho) R_i(\rho) \right) M_i^{-T}(\rho),
\end{align}
where $M_i(\rho) N_i^T(\rho) = I - R_i(\rho) S_i(\rho)$. The matrices $F_i(\rho)$ and $L_i(\rho)$ are defined by
\begin{align*}
F_i(\rho) & = -\left( D_{12,i}^T(\rho) D_{12,i}(\rho) \right)^{-1} \left[ \gamma_i B^T_{2,i}(\rho)
R_i^{-1}(\rho) + D_{12,i}^T(\rho) C_{1,i}(\rho) \right] \\
L_i(\rho) & = -\left[ \gamma_i S_i^{-1}(\rho) C^T_{2,i}(\rho) +
B_{1,i}(\rho) D_{21,i}^T(\rho) \right] \left( D_{21,i}(\rho) D^T_{21,i}(\rho) \right)^{-1}.
\end{align*}

The key idea behind the proposed hybrid LPV control
framework is to combine locally optimized LPV controllers
with a hysteresis-based switching mechanism while
preserving global stability.
Parameter-dependent Lyapunov functions guarantee
exponential stability and performance within each parameter  subset, whereas the hysteresis switching logic enforces
non-increasing Lyapunov values at switching instants.
The controller state reset mechanism bridges Lyapunov
functions across adjacent regions, converting otherwise
non-convex boundary conditions into tractable LMIs.
As a result, the method achieves reduced conservatism and
improved local performance without sacrificing stability
under parameter variations and controller switching.

%   By limiting parameter variation within each subregion, the %   hybrid LPV controller achieves reduced conservatism and
%   improved local performance, while hysteresis switching
%   ensures global stability.
Different performance levels are permitted across parameter subsets, allowing the controller gains to adapt to local operating conditions and providing a flexible framework for high-performance LPV control.
Specifically, the hybrid LPV control synthesis conditions in Theorem \ref{thm:hyst_syn} can be formulated as a convex optimization problem of the form
%   The parameter set ${\cal P}$ is partitioned into multiple
%   subregions ${\cal P}_1, {\cal P}_2, \ldots, {\cal
%   P}_{N_p}$.
%   For each subset, a local LPV controller is synthesized
%   according to Theorem \ref{thm:hyst_syn}.
%   The global gain-scheduled controller is then constructed
%   by switching among these local controllers as $\rho$
%   transitions between subsets.
\begin{align*}
\min_{R_,S_i,\hat{\Delta}_{ij},\gamma_i} & \sum_{i=1}^{N_p} \alpha_i \gamma_i \\ \mbox{subj. to} \quad & (\ref{eqn:hystLMI1})-(\ref{eqn:hystLMI4})
\end{align*}
where $\alpha_i, i \in Z_{N_p}$ are weighting coefficients that specify the relative performance priorities of the individual parameter subsets.

%	The solvability conditions provided in
%	(\ref{eqn:lpv_lmi1})--(\ref{eqn:lpv_lmi3}) are clearly
%	infinite-dimensional, as is the solution space. To approximate, we
%	restrict the search of parameter-dependent Lyapunov function to a
%	span of finite number of basis functions. That is, let
%	\[
%	X^k(\rho) = \sum_{i=1}^{N_f} f_i^k(\rho) X_i^k, \qquad Y^k(\rho) =
%	\sum_{j=1}^{N_g} g_j^k(\rho) Y_j^k
%	\]
%	where $\left\{ f_i^k(\rho) \right\}_{i=1}^{N_f}$ and $\left\{
%	g_j^k(\rho) \right\}_{i=1}^{N_g}$ are user-specified scalar basis
%	functions. $X_i^k, Y_j^k$ are new optimization variables to be
%	determined. In particular, different basis functions can be used
%	over different parameter subsets with the continuity constraints
%	satisfied at the boundary. After such a parameterization,
%	the LPV synthesis conditions can be solved using a gridding method
%	over the entire parameter space.

\section{AMB Control Case Study}
\label{sec:amb}

This section applies the proposed hybrid LPV control methodology to an active magnetic bearing (AMB) system and demonstrates its advantages over conventional LPV designs.

\subsection{Rotor Dynamics Modeling}
\label{subsec:model}

Due to the linear dependence of the system dynamics on rotor speed, the nonlinear gyroscopic equations governing AMBs can be simplified into a set of linear time-varying differential equations. Nevertheless, AMB rotor dynamics are inherently unstable: even small mass imbalances can generate large synchronous disturbances at the rotor's rotational frequency. Consequently, a gain-scheduled control strategy is required to ensure effective disturbance rejection over a wide operating speed range.

%	The rotational motion of a magnetic bearing can be derived from its rigid
%	body dynamics, which is \cite{MohB95,TsiM96}
%	\begin{align*}
%	\ddot{\theta} & = -\frac{\rho J_a}{J_r} \dot{\psi} + \frac{\ell}{J_r}
%	(f_{r1} - f_{r2} + f_{\ell 2} - f_{\ell 1} + f_{d \theta}) \\
%	\ddot{\psi} & = \frac{\rho J_a}{J_r} \dot{\theta} + \frac{\ell}{J_r}
%	(f_{r3} - f_{r4} + f_{\ell 4} - f_{\ell 3} + f_{d \psi})
%	\end{align*}
%	where $\theta, \psi$ are the Euler angles denoting the orientation of
%	rotor centerline. $J_a, J_r$ are the moment of inertia of the rotor in
%	axial and radial directions, respectively. $\rho$ denotes the rotor
%	speed. The magnetic forces generated by four pairs of electromagnets are
%	denoted by $f_{ri}, f_{\ell i}$ for $i = 1, 2, 3, 4$. $f_{d \theta},
%	f_{d \psi}$ are disturbance forces caused by gravity, modeling errors,
%	imbalances, etc.
%
%	The electromagnetic force $f_j$ is related to the voltage $e_j$ across
%	the $jth$ coil through the magnetic flux $\psi_j$ by the equations
%	\begin{align*}
%	f_j = k \phi_j^2 \left( 1 + \frac{2 g_j}{\pi h} \right) \\
%	e_j = N \frac{d \phi_j}{d t} + \frac{2 R}{\nu_0 A N} g_j \phi_j
%	\end{align*}

The complete rigid-body dynamic model of the AMB system was derived in \cite{MohB95,TsiM96} and is given by
\begin{align}
\ddot{\theta} & = -\frac{\rho J_a}{J_r} \dot{\psi} + \frac{\ell}{J_r}
(-4 c_2 \ell \theta + 2 c_1 \phi_{\theta} + f_{d \theta}) \\
\ddot{\psi} & = \frac{\rho J_a}{J_r} \dot{\psi} + \frac{\ell}{J_r}
(-4 c_2 \ell \phi + 2 c_1 \phi_{\psi} + f_{d \psi} ) \\
N \dot{\phi}_{\theta} & = e_{\theta} + 2 d_2 \ell \theta - d_1 \phi_{\theta} \\
N \dot{\phi}_{\psi} & = e_{\psi} + 2 d_2 \ell \psi - d_1 \phi_{\psi},
\end{align}
where $\theta$ and $\psi$ denote the Euler angles describing the rotor centerline orientation; $J_a$ and $J_r$ are the axial and radial moments of inertia, respectively; and $\rho$ denotes the rotor speed. The variables $\phi_{\theta}$ and $\phi_{\psi}$ represent the differential magnetic fluxes generated by the electromagnetic actuator pairs, while $e_{\theta}$ and $e_{\psi}$ are the corresponding control voltages. The constants $c_1$, $c_2$, $d_1$, and $d_2$ depend on $\Phi_0$, $G_0$, $R$, $A$, $N$, $\nu_0$, and the bearing geometry, as defined in
\begin{align*}
c_1 & = 2 k \Phi_0 \left( 1 + \frac{2 G_0}{\pi h} \right), \qquad
c_2 = \frac{2 k \Phi_0^2}{\pi h}, \\
d_1 & = \frac{2 R G_0}{\nu_0 A N}, \qquad d_2 = \frac{2 R \Phi_0}{\nu_0 A N}
\end{align*}

The system parameters used in this study are summarized in Table I.
\begin{table}[htb]
\begin{center}
\begin{tabular}{c|c} \hline
Parameter & Value \\ \hline
$A$ area of each pole ($mm^2$) & $1531.79$ \\
$h$ pole width ($mm$) & $40.00$ \\
$G_0$ nominal gap($mm$) & $0.55$ \\
$J_r$ radial moment of inertia ($kg \cdot m^2$) & $0.333$ \\
$J_a$ axial moment of inertia ($kg \cdot m^2$) & $0.0136$ \\
$\ell$ half the length of the shaft ($m$) & $0.13$ \\
$k$ & $4.6755576 \times 10^8$ \\
$N$ number of coil turns & $400$ \\
$R$ coil resistance ($Ohm$) & $14.6$ \\
$\Phi_0$ nominal airgap ($Wb$) & $2.09 \times 10^{-4}$ \\ \hline
\end{tabular} \\
\caption{Active magnetic bearing parameters.}
\end{center}
\end{table}

Rotor imbalance forces $f_{d\theta}$ and $f_{d\psi}$ are modeled as sinusoidal disturbances:
\begin{align*}
f_{d \theta} &= \frac{(J_r - J_a)}{\ell} \rho^2 \tau \cos(\rho t), \\
f_{d \psi} &= \frac{(J_r - J_a)}{\ell} \rho^2 \tau \sin(\rho t).
\end{align*}
For automatic balancing design, these forces are treated as sinusoidal sensor noise
$n^T = \left[ \ell \tau \sin (\rho t + \lambda) \right.$
$\left. \ \ \ell \tau \cos (\rho t + \lambda) \right]$, which can be represented in state-space form as
\begin{align}
\dot{x}_{di} & = \begin{bmatrix} 0 & -\rho \\ \rho & 0 \end{bmatrix} x_{di} +
\begin{bmatrix} 1 \\ 0 \end{bmatrix} \tilde{d} \\
n & = \begin{bmatrix} 1 & 0 \\ 0 & 1 \end{bmatrix} x_{di}
\end{align}

Let $x^T = \left[ \ell \theta \ \ \ell \psi \ \ \ell \dot{\theta} \ \
\ell \dot{\psi} \ \ \phi_\theta \ \ \phi_\psi \ \ x_{d1} \ \ x_{d2}
\right]$, $u^T = \left[ e_{\theta} \ \
e_{\psi} \right]$ and $w = \tilde{d}$. Combining the rotor dynamics with the disturbance model yields the linearized state-space representation
\begin{align*}
\dot{x} & = A(\rho) x + B_1 w + B_2 u \\
z & = C_1 x + D_{11} w + D_{12} u \\
y & = C_2 x + D_{21} w + D_{22} u
\end{align*}
with
\begin{align*}
A(\rho) & = \begin{bmatrix}
0 & 0 & 1 & 0 & 0 & 0 & 0 & 0 \\
0 & 0 & 0 & 1 & 0 & 0 & 0 & 0 \\
-\frac{4 c_2}{m} & 0 & 0 & -\frac{\rho J_a}{J_r} & \frac{2 c_1}{m} & 0 & 0 & 0
\\
0 & -\frac{4 c_2}{m} & \frac{\rho J_a}{J_r} & 0 & 0 & \frac{2 c_1}{m} & 0 & 0 \\
\frac{2 d_2}{N} & 0 & 0 & 0 & -\frac{d_1}{N} & 0 & 0 & 0 \\
0 & \frac{2 d_2}{N} & 0 & 0 & 0 & -\frac{d_1}{N} & 0 & 0 \\
0 & 0 & 0 & 0 & 0 & 0 & 0 & -\rho \\
0 & 0 & 0 & 0 & 0 & 0 & \rho & 0
\end{bmatrix}, \\
B_1 & = \begin{bmatrix}
0_{6 \times 1} \\
\begin{matrix} 1 \\ 0 \end{matrix}
\end{bmatrix}, \qquad
B_2 = \frac{1}{N} \begin{bmatrix}
0_{4 \times 2} \\
I_2 \\
0_{2 \times 2}
\end{bmatrix} \\
C_1 & = \begin{bmatrix}
\begin{matrix} I_2 & 0_{2 \times 6} \end{matrix} \\
0_{2 \times 8}
\end{bmatrix}, \qquad C_2 = \begin{bmatrix}
I_2 & 0_{2 \times 4} & I_2
\end{bmatrix} \\
D_{11} &= 0_{4 \times 1}, \qquad
D_{12} = \begin{bmatrix}
0_{2 \times 2} \\
I_2
\end{bmatrix} \\
D_{21} &= 0_{2 \times 1}, \qquad
D_{22} = 0_{2 \times 2}.
\end{align*}
The resulting model is an LPV system, with rotor speed $\rho$ serving as the scheduling parameter. For gain-scheduled control, $\rho$ is assumed to be measurable in real time.

\subsection{Hybrid LPV Controller Design and Simulation}
\label{subsec:lpvcontr}

The objective of the hybrid LPV controller design is to asymptotically stabilize the AMB system over the entire operating speed range while minimizing a performance output defined as a weighted combination of bearing forces, rotor displacements at bearing locations, and control effort.

Although the scheduling parameter varies with time in the LPV model, it is treated as a frozen parameter during controller synthesis. Performance objectives are enforced through weighting functions, and the resulting weighted open-loop interconnection is shown in Fig. \ref{fig:wtop},
\begin{figure}[!htb]
\centering
\includegraphics[width=0.5\hsize]{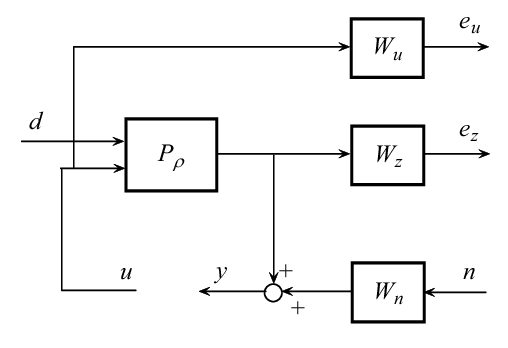}
%   \input{texfig/wtop.tex}
%   \end{center}
\caption{Weighted open-loop interconnection for the magnetic bearing
system.}
\label{fig:wtop}
\end{figure}
where the selected weighting functions are given by
\begin{align*}
W_e(s) & = \frac{30 (s + 8)}{s + 0.001} I_2 \\
W_u(s) & = \frac{0.01 (s+100)}{s + 100000} I_2, \\
W_n(s) &= 0.001 I_2.
\end{align*}

The rotor speed is assumed to vary within the range $300-2000$ rad/s, over which significant gyroscopic effects are observed. To enable multiple-region LPV synthesis, the parameter domain is first partitioned into two overlapping subregions, $[300,1200]$ rad/s and $[1100,2000]$ rad/s, with ten uniformly spaced grid points in each region.

An LPV controller first synthesized using a single parameter-dependent quadratic Lyapunov function (PDLF) over the entire domain \cite{WuYPB96} yields a performance bound from  $0.3472$ to $0.7799$ for different parameter variation rates, indicating substantial conservatism.
The proposed hybrid LPV approach employs two identical basis functions within each subregion to parameterize the Lyapunov matrices:
\begin{align*}
f_1^k(\rho) & = 1, \\
g_1^k(\rho) & = 1, \qquad g_2^k(\rho) = \rho
\end{align*}
for $k = 1,\ldots,N_p$, where $R(\rho)=\sum_{i=1}^{n_f} f_i^k R_i^k$ and $S(\rho)=\sum_{i=1}^{n_f} f_i^k S_i^k$.
Within each subregion, $R(\rho)$ is chosen to be constant, eliminating controller gain dependence on the parameter rate $\dot{\rho}$.
Since the AMB model contains only one scheduling parameter, the switching matrices $\Delta_{ij}$ are specified as constants at the switching boundaries.
The resulting performance levels under different assumptions on parameter variation rates are summarized in Table \ref{tab:PDLFresults1}.
Since a conventional LPV controller is designed to cover the entire parameter space, it is inherently suboptimal at individual frozen parameter values.
In contrast, the proposed hybrid LPV controller achieves performance comparable to that of a single PDLF-based design overall, while providing improved performance in selected operating subregions, particularly at higher rotor speeds.

\begin{table}[htb]
\begin{center}
\begin{tabular}{c|cc} \hline
Parameter variation rate (rad/s) & Hybrid LPV & LPV \\ \hline
$\left[ 0.1 \ \ 0.1 \right]$ & $[0.3687 \ \ 0.2982]$ & 0.3472 \\
$\left[ 5 \ \ 5 \right]$ & $[0.3698 \ \ 0.2976]$ & 0.3496 \\
$\left[ 50 \ \ 50 \right]$ & $[0.4090 \ \ 0.2891]$ & 0.3796 \\
$\left[ 300 \ \ 300 \right]$ & $[0.8590 \ \ 0.3946]$ & 0.7799 \\ \hline
\end{tabular}
\caption{${\cal L}_2$ gain performance using single PDLF and multiple PDLFs over two regions.}
\label{tab:PDLFresults1}
\end{center}
\end{table}

To further investigate scalability, the parameter space is also partitioned into four overlapping subregions: $[300,800]$, $[700,1200]$, $[1100,1600]$, and $[1500,2000]$ rad/s, with six grid points per region.
The same basis functions are used for fair comparison, and the results are reported in Table \ref{tab:PDLFresults2}.
It is clear that increase parameter subsets help improving local controlled performance.

\begin{table}[htb]
\begin{center}
\begin{tabular}{c|cc} \hline
Parameter variation rate (rad/s) & hybrid LPV \\ \hline
$\left[ 0.1 \ \ 0.1 \ \ 0.1 \ \ 0.1 \right]$ & $[0.5674 \ \ 0.1866 \ \ 0.1817 \ \ 0.1816]$ \\
$\left[ 5 \ \ 5 \ \ 5 \ \ 5 \right]$ & $[0.5696 \ \ 0.1856 \ \ 0.1808 \ \ 0.1805]$ \\
$\left[ 50 \ \ 50 \ \ 50 \ \ 50 \right]$ & $[0.5781 \ \ 0.1933 \ \ 0.1875 \ \ 0.1867]$ \\
$\left[ 300 \ \ 300 \ \ 300 \ \ 300 \right]$ & $[0.9471 \ \ 0.4097 \ \ 0.3416 \ \ 0.3223]$ \\ \hline
\end{tabular}
\caption{${\cal L}_2$ gain performance using multiple PDLFs over four regions.}
\label{tab:PDLFresults2}
\end{center}
\end{table}

Finally, nonlinear simulations are conducted using a time-varying rotor speed profile:
\[
\rho(t) = \left\{ \begin{tabular}{ll}
$650$ rad/s & $0 < t < 0.5 $ s \\
$650 + 50 (t-0.5)$ rad/s & $0.5 \leq t < 13.5 $ s \\
$1300$ rad/s & $13.5 \leq t < 16 $ s \\
$1300 - 50 (t-16)$ rad/s & $16 \leq t < 22 $ s\\
$1000$ rad/s & $t \geq 22$ s
\end{tabular} \right.
\]
This profile intentionally crosses scheduling boundaries to illustrate controller switching behavior.
The disturbance magnitude is set to $\tilde d = 1.3 \times 10^{-4}$.
The hybrid LPV controller with four parameter sub-regions are examined.
Controller switches occur at $t=3.5$ s and $t=11.5$ s during acceleration, corresponding to $\rho=800$ and $1200$ rad/s, and at $t=20$ s during deceleration at $\rho=700$ rad/s.

The nonlinear simulation results with the corresponding controller index can be found in Figs. \ref{fig:hybsimu}-\ref{fig:ctrlseq}.
Aside from small transients during switching, simulation results demonstrate performance comparable to that of a single LPV controller. The hybrid LPV design effectively exploits speed-dependent tuning, achieving improved performance under time-varying operating conditions.

\begin{figure}[!htb]
\begin{minipage}[b]{3.5in}
%   \centering\epsfig{file=epsfile/hyb_x.eps,height=2.0in,width=3.0in} \\
\includegraphics[width=0.8\hsize]{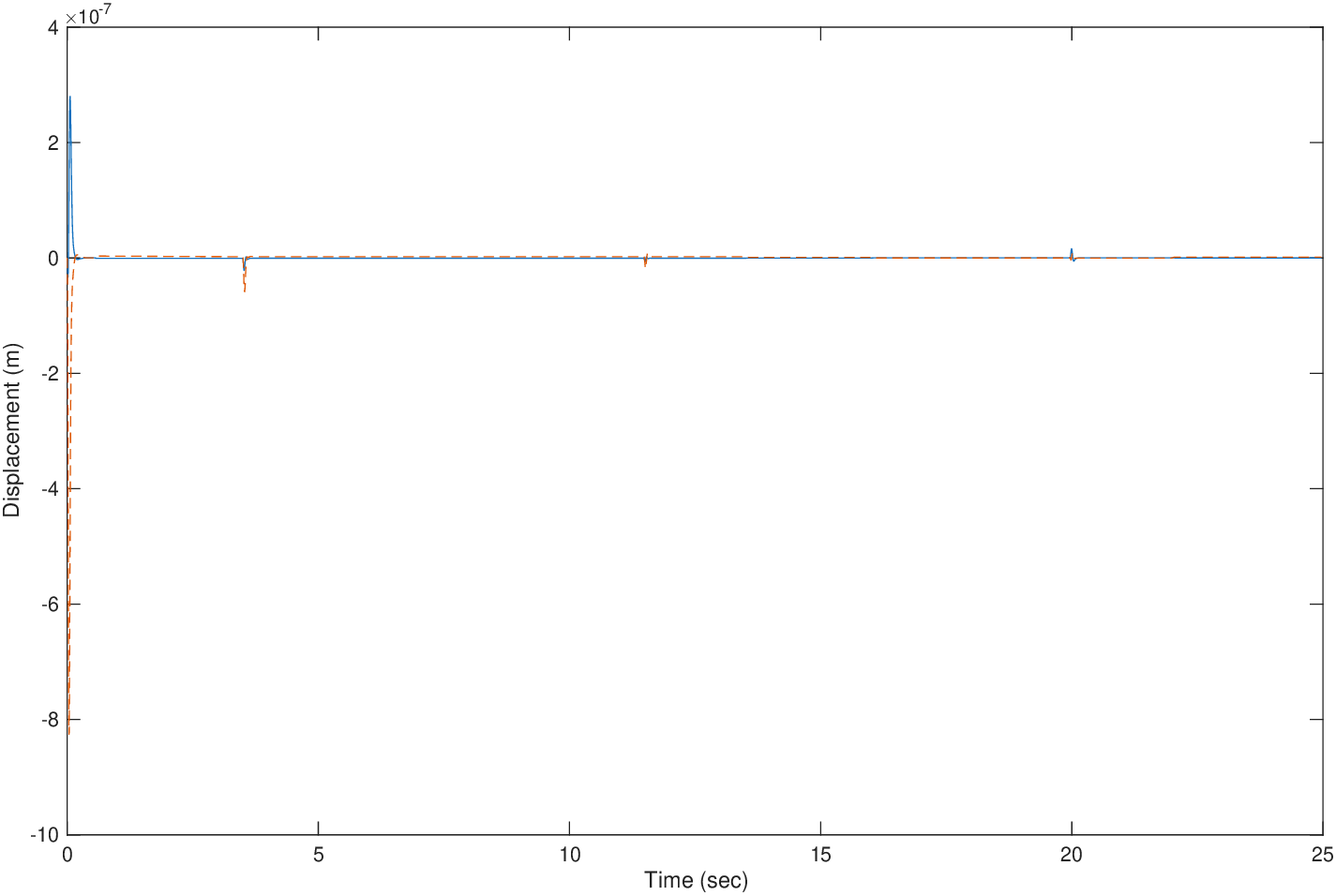} \\
\centering{(a). rotor displacement: $x_1$--solid, $x_2$--dash.}
\end{minipage}
\hfill
\begin{minipage}[b]{3.5in}
%   \centering\epsfig{file=epsfile/hyb_u.eps,height=2.0in,width=3.0in} \\
\includegraphics[width=0.8\hsize]{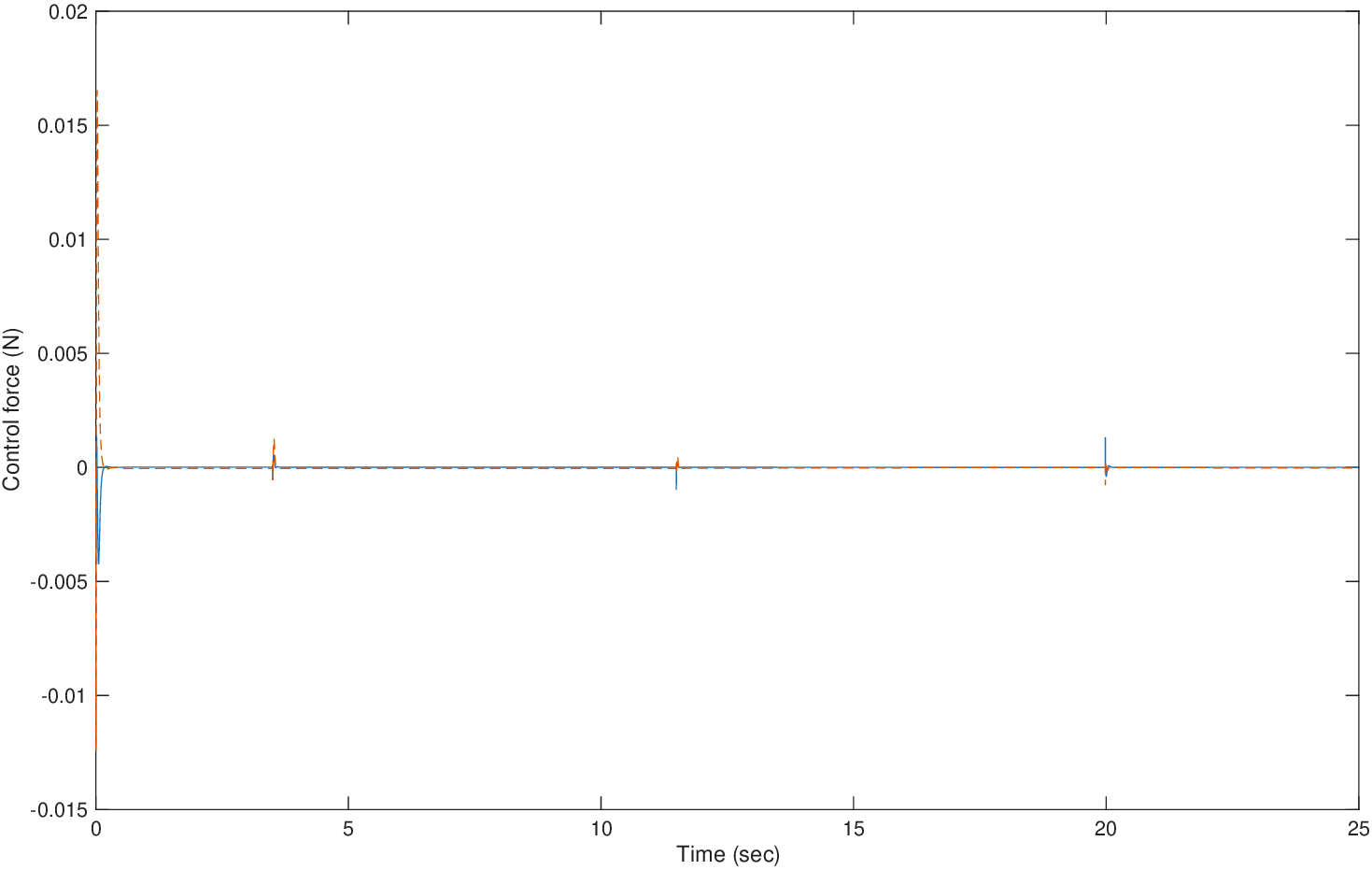} \\
\centering{(b). control force: $u_1$--solid, $u_2$--dash.}
\end{minipage}
\caption{The performance of hybrid LPV controller for a time-varying rotor speed profile.}
\label{fig:hybsimu}
\end{figure}

\begin{figure}[!htb]
\centering
\includegraphics[width=0.5\hsize]{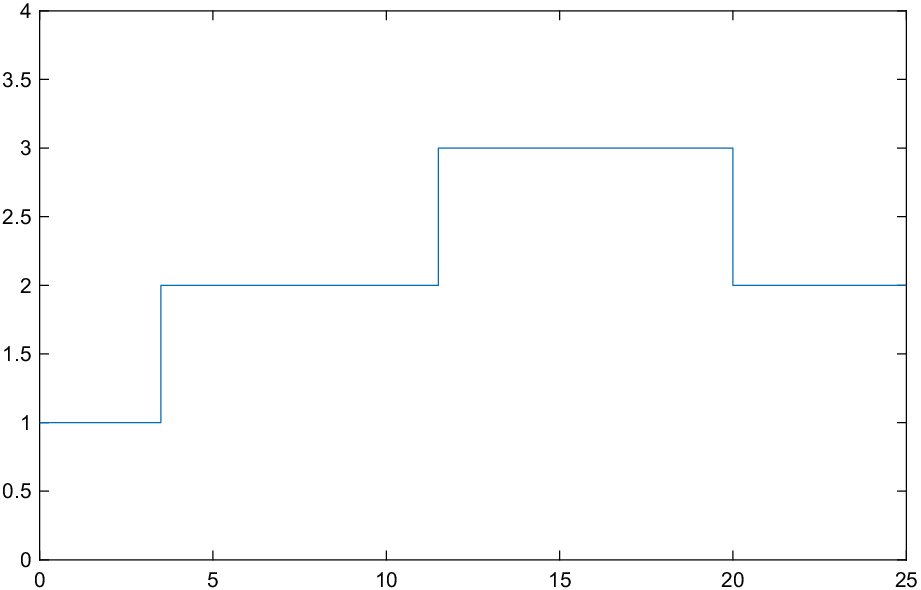}
\caption{LPV controller activation sequence.}
\label{fig:ctrlseq}
\end{figure}

\section{Concluding Remarks}
\label{sec:conclusion}

This paper has presented a hybrid LPV control framework with hysteresis switching logic for high-performance gain-scheduled control of systems with large parameter variations.
By introducing a controller state-reset mechanism and enforcing Lyapunov function monotonicity across parameter subset boundaries, closed-loop stability under switching is guaranteed without requiring Lyapunov continuity.
The resulting synthesis conditions are formulated as convex
LMI optimization problems, allowing efficient and simultaneous computation of both local LPV controllers and reset matrices.

The proposed approach was applied to an active magnetic bearing (AMB) system, whose rotor dynamics exhibit strong dependence on rotational speed.
By partitioning the parameter space and switching among locally optimized LPV controllers, the hybrid design reduced conservatism associated with conventional LPV methods.
Simulation results demonstrate effective rejection of synchronous imbalance disturbances and improved performance over a wide operating speed range, particularly in high-speed regimes.

Future research will focus on extending the proposed framework to flexible-mode AMB models, where structural resonances,
unmodeled dynamics, and uncertainty become increasingly
important at high rotational speeds \cite{BalWS2012}. 
The integration of robust and hybrid gain-scheduling
techniques is expected to further enhance performance,
stability margins, and reliability in advanced magnetic bearing applications.


\begin{thebibliography}{99}

\bibitem{Ant.IEEE00}
P.~J. Antsaklis,
``A brief introduction to the theory and applications of
hybrid systems,''
\emph{Proc. {IEEE}}, 88(7):879-887, 2000.

\bibitem{AtoSMM2022}
H. Atoui, O. Sename, V. Milanes, and J.J. Martinez,
``Smooth switching of multi-LPV control systems based-on Youla-Kucera parametrization,''
{\em IFAC PaperOnLine}, 55-35:19-24, 2022.

%   \bibitem{ApkA97}
%   P. Apkarian and R.J. Adams,
%   ``Advanced Gain-scheduling Techniques for Uncertain
%   Systems,''
%   {\em IEEE Trans. Contr. Syst. Tech.}, {\bf 6}(1):21-32,
%   1997.

%   \bibitem{BecP94}
%   G.~Becker and A.K. Packard,
%   ``Robust performance of linear parametrically varying
%   systems using parametrically dependent linear dynamic
%   feedback,''
%   {\em Syst. Contr. Letts.}, 23(3):205-215,1994.

\bibitem{BalWS2012}
H.M.N.K. Balini, J. Witte, and C.W. Scherer,
``Synthesis and implementation of gain-scheduling and LPV controllers for an AMB system,''
{\em Automtica}, 48:512-517, 2012.

\bibitem{BoyEFB94}
S.P. Boyd, L.~{El Ghaoui}, E.~Feron, and V.~Balakrishnan,
{\em Linear Matrix Inequalities in Systems and Control
Theory},
Philadelphia, {PA}:SIAM, 1994.

\bibitem{Bra.TAC98}
M.~S. Branicky,
``Multiple {Lyapunov} functions and other analysis tools for
  switched and hybrid systems,''
\emph{{IEEE} Trans. Automat. Control}, 43(4):475-482, 1998.

\bibitem{BraBM.TAC98}
M.~S. Branicky, V.~S. Borkar, and S.~K. Mitter,
``A unified framework for hybrid control: Model and optimal control theory,''
\emph{{IEEE} Trans. Automat. Control}, 43(1):31-45, 1998.

%   \bibitem{FeuGS.ACC97}
%   A.~Feuer, G.~C. Goodwin, and M.~Salgado,
%   ``Potential benefits of hybrid control for linear time
%   invariant plants,''
%   in \emph{Proc. of American Control Conference}, pp. 2790-2794, 1997.

%   \bibitem{FitK99}
%   R. Fittro and C.R. Knospe,
%   ``$\mu$-control of a high speed spindle thrust magnetic
%   bearing,''
%   in {\em Proc. 6th IEEE Conf. Contr. Appli.}, 1999.

\bibitem{FujHM93}
M. Fujita, K. Hatake, and F. Matsumura,
``Loop shaping based robust control of a magnetic bearing,''
{\em IEEE Contr. Syst. Magazine}, 13(1):57-65, 1993.

%   \bibitem{GahNLC95}
%   P. Gahinet, A. Nemirovskii, A. Laub, and M. Chilali,
%   {\em LMI Control Toolbox},
%   Mathworks Inc., Natick, {MA}, 1995.

\bibitem{GoeST.B12}
R.~Goebel, R.G. Sanfelice, and A.R. Teel,
\emph{Hybrid Dynamical Systems},
Princeton, NJ: Princeton University Press, 2012.

%   \bibitem{HesM.CDC99}
%   J.~P. Hespanha and A.~S. Morse,
%   ``Stability of switched systems with average dwell-time,''
%   in \emph{Proc. 38th. {IEEE} CDC.}, pp. 2655-2660, 1999.

\bibitem{HugW.B12}
H. Hughes and F. Wu,
``LPV ${\cal H}_\infty$ Control of Flexible Hypersonic Vehicle,'' in {\em Control of Linear Parameter Varying Systems with Applications} (J. Mohammadpour and C.W. Scherer Ed.), Springer, 2012. DOI: 10.1007/978-1-4614-1833-7\_16

%   \bibitem{KolM.TAC96}
%   I.~Kolmanovsky and N.~H. McClamroch,
%   ``Hybrid feedback laws for a class of cascade nonlinear
%   control systems,''
%   \emph{{IEEE} Trans. Autom. Control}, 41(9):1271-1282, 1996.

\bibitem{LiberzonB03}
D.~Liberzon,
\emph{Switching in Systems and Control},
Boston, MA: Birkhauser, 2003.

\bibitem{LibM.CSM99}
D.~Liberzon and A.~S. Morse,
``Basic problems in stability and design of switched systems,'' \emph{IEEE Contr. Syst. Magazine}, 19(5):59-70, 1999.

\bibitem{LinA.TAC09}
H.~Lin and P.~J. Antsaklis,
``Stability and stabilizability of switched linear systems: A survey of recent results,''
\emph{{IEEE} Trans. Automat. Control}, 54(1):308-322, 2009.

\bibitem{LuW.ACC04}
B. Lu and F. Wu, 
``Control design of switched LPV systems using multiple parameter-dependent Lyapunov functions,'' 
in {\em Proc. Amerian Control Conference}, pp. 3875-3880, 2004. DOI: 10.23919/ACC.2004.1384517

\bibitem{LuW.Au04}
B.~Lu and F.~Wu,
``Switching {LPV} control designs using multiple parameter-dependent {Lyapunov} functions,''
\emph{Automatica}, 40:1973-1980, 2004.

%   \bibitem{LuW2006}
%   B. Lu and F. Wu,
%   ``Switching LPV control designs using multiple
%   parameter-dependent Lyapunov functions,''
%   {\em Automatica}, 40:1973-1980, 2004.

%	\bibitem{MasTA98}
%	S. Mason, P. Tsiotras and P. Allaire.
%	``Linear Parameter Varying Controllers for Flexible Rotors Supported
%	on Magnetic Bearings,''
%	in {\em Proc. 6th Int. Sym. Magnetic Bearings}, 1998.

\bibitem{MatNHF96}
F. Matsumura, T. Namerikawa, K. Hagiwara, and M. Fujita,
``Application of gain scheduled ${\cal H}_\infty$ robust controllers to a magnetic bearing,''
{\em IEEE Trans. Contr. Syst. Techno.}, 4(5):484-493, 1996.

\bibitem{McCK.IEEE00}
N.~H. McClamroch and I.~Kolmanovsky,
``Performance benefits of hybrid control design for linear and nonlinear systems,''
\emph{Proc. of IEE}, 88(7):1083-1096, 2000.

\bibitem{MohB95}
A.M. Mohamed and I. Busch-Vishmiac,
``Imbalance compensation and automation balancing in magnetic bearing
systems using the $Q$-parameterization theory,''
{\em IEEE Trans. Contr. Syst. Techno.}, 3(2):202-211, 1995.

\bibitem{Mor.B97}
A.~S. Morse,
\emph{Control using Logic-based Switching}, Heidelberg, Germany: Springer, 1997.

\bibitem{MorMG92}
A.S. Morse, D.Q. Mayne, and G.C. Goodwin,
``Applications of hysteresis switching in parameter adaptive control,''
{\em IEEE Trans. Automat. Control}, 37(9):1343-1354, 1992.

\bibitem{NonI96}
K. Nonami and T. Ito,
``$\mu$ synthesis of flexible rotor-magnetic bearing system,''
{\em IEEE Trans. Contr. Syst. Techno.}, 4(5):503-512, 1996.

%   \bibitem{Pac94}
%   A.K. Packard,
%   ``Gain Scheduling via Linear Fractional Transformations,''
%   {\em Syst. Contr. Letts.}, {\bf 22}(2):79--92, 1994.

%	\bibitem{Rug91}
%	W.J. Rugh.
%	\newblock Analytical framework for gain scheduling.
%	\newblock {\em {IEEE} Contr. Sys. Mag.}, {\bf 11}(1):74--84, 1991.

\bibitem{SchS.B00}
A.~van~der Schaft and H.~Schumacher,
\emph{An Introduction to Hybrid Dynamical Systems},
Springer, 2000.

%   \bibitem{Sch99}
%   C.W. Scherer,
%   ``Robust Mixed Control and LPV Control with Full Block Scalings,''
%   in {\em Recent Advances of LMI Methods in Control} (L. El
%   Ghaoui, S. Niculescu ed.), SIAM, 1999.

\bibitem{SchBT94}
G. Schweitzer, H. Bleuler, and A. Traxler,
{\em Active Magnetic Bearings: Basics, Properties and Applications},
Verlag vdf Hochschulverlag AG and der ETH, Zurich Germany, 1994.

\bibitem{ShoWMWK2007}
R. Shorten, F. Wirth, O. Mason, K. Wulff, and C. King, ``Stability criteria for switched and hybrid systems,''
{\em SIAM Review}, 49(4):545-592, 2007.

\bibitem{SivN96}
S. Sivrioglu and K. Nonami,
``LMI approach to gain scheduled ${\cal H}_{\infty}$ control beyond
PID control for gyroscopic rotor-magnetic bearing systems,''
in {\em Proc. 35th IEEE Conf. Dec. Contr.}, pp. 3694-3699, 1996.

\bibitem{SunG.B05}
Z.~Sun and S.~S. Ge,
\emph{Switched Linear Systems: Control and Design}, Verlag, NY: Springer, 2005.

%   \bibitem{TsiK97}
%   P. Tsiotras and C. Knospe,
%   ``Reducing Conservatism for Gain-scheduled ${\cal
%   H}_\infty$ Controllers for AMB's,''
%   in {\em Proc. MAG'97 Magnetic Bearings Ind. Conf. Exhib.}, %   1997.

\bibitem{TsiM96}
P. Tsiotras and S. Mason.
``Self-scheduled ${\cal H}_\infty$ controllers for magnetic bearings,''
in {\em Proc. ASME Int. Mech. Eng. Congress Expo.}, pp.~151-158, 1996.

%	\bibitem{Van88}
%	J.M. Vance,
%	{\em Rotordynamics of Turbomachinery},
%	Wiley Interscience, New York, NY, 1988.

%	\bibitem{Wu95}
%	F. Wu,
%	{\em Control of Linear Parameter Varying Systems},
%	Ph.D. Thesis, University of California at Berkeley, 1995.

\bibitem{Wu2001}
F. Wu,
``Switching LPV control design for magnetic bearing systems,'' in {\em Proc. IEEE Int. Conference on Control Applications}, pp. 41-47, 2001. DOI: 10.1109/CCA.2001.973835

\bibitem{WuYPB96}
F. Wu, X.H. Yang, A.K. Packard, and G. Becker.
``Induced ${\cal L}_2$ norm control for {LPV} systems with
bounded parameter variation rates,''
{\em Int. J. Robust Non. Contr.}, 6(9/10):983-998,
1996.

\bibitem{YeMH.TAC98}
H.~Ye, A.~N. Michel, and L.~Hou,
``Stability theory for hybrid dynamical systems,''
\emph{{IEEE} Trans. Automat. Control}, 43(4):461-474, 1998.

\bibitem{YuaW2015}
C. Yuan and F. Wu,
``Hybrid control for switched linear systems with
average swell time,''
{\em IEEE Trans. Automat. Control}, 60(1):240-245, 2015.

\bibitem{YuaLWD2016}
C. Yuan, Y. Liu, F. Wu, and C. Duan,
``Hybrid switched gain-scheduling control for missile
autopilot design,''
{\em J. of Guiance, Control, And Dynamics}, 39(10):2352-2363, 2016. DOI: 10.2514/1.G001791

\end{thebibliography}
\end{document}